\documentclass[10pt, journal, letterpaper]{IEEEtran}
\usepackage{graphicx}
\usepackage{epsfig}
\usepackage{subfigure}
\usepackage{tabularx}
\usepackage{amssymb}
\usepackage{latexsym}
\usepackage{amsmath}
\usepackage{algorithm}
\usepackage[noend]{algpseudocode}
\usepackage{color}
\usepackage{relsize}
\usepackage{amsthm}
\usepackage[dvipsnames]{xcolor}
\usepackage[,skip=3pt]{caption}
\usepackage{cite}
\newtheorem{theorem}{Theorem}[]

\theoremstyle{definition}
\newtheorem{definition}{Definition}	
\newtheorem{observation}{Observation}
\usepackage{esvect}
\theoremstyle{remark}
\begin{document}
	
\title{3D Placement and Orientation of mmWave-based UAVs for Guaranteed LoS Coverage}

\author{Javad Sabzehali, Vijay K. Shah, Harpreet S. Dhillon, and Jeffrey H. Reed
\thanks{The authors are with Wireless@VT, The Bradley Department of ECE at Virginia Tech, Blacksburg, VA (emails: \{jsabzehali, vijays, hdhillon, reedjh\}@vt.edu). This work is partially supported by NSF Grant CNS-1923807 and Commonwealth Cyber Initiative (CCI), an investment in the advancement of cyber R\&D, innovation, and workforce development. For more information about CCI, visit {www.cyberinitiative.org}}
}

\maketitle

\begin{abstract} 
Unmanned aerial vehicles (UAVs), as aerial base stations, are a promising solution for providing wireless communications, thanks to their high flexibility and autonomy. Moreover, emerging services, such as extended reality, require high-capacity communications. To achieve this, millimeter wave (mmWave), and recently, terahertz bands have been considered for UAV communications. However, communication at these high frequencies requires a line-of-sight (LoS) to the terminals, which may be located in 3D space and may have extremely limited direct-line-of-view (LoV) due to blocking objects, like buildings and trees. In this paper, we investigate the problem of determining 3D placement and orientation of UAVs such that users have guaranteed LoS coverage by at least one UAV and the signal-to-noise ratio (SNR) between the UAV-user pairs are maximized. We formulate the problem as an integer linear programming(ILP) problem and prove its NP-hardness. Next, we propose a low-complexity geometry-based greedy algorithm to solve the problem efficiently. Our simulation results show that the proposed algorithm (almost) always guarantees LoS coverage to all users in all considered simulation settings.

\end{abstract}


\maketitle
\section{Introduction}
Unmanned Aerial Vehicles (UAVs), also known as drones, have numerous applications in telecommunications, rescue operations, aerial sensing, to name a few. Significant research efforts have gone into understanding the role of UAVs as aerial base stations (BSs) to complement the coverage and capacity of terrestrial wireless networks. mmWave and future teraHertz communications with orders of magnitude larger bandwidths than the conventional sub-6 GHz systems are expected to play a pivotal role in supporting emerging applications, such as extended reality, in UAV-assisted communications networks \cite{tripathi2021millimeterwave}. However, characterized by very high path loss in non-line-of-sight (NLoS) conditions, \textit{it is critical that mmWave (and teraHertz) based UAV communications always have a line-of-sight (LoS) coverage to wireless devices~} \cite{8984705}\cite{8338071}.

To provision high quality wireless LoS coverage to wireless devices (or users), UAVs should jointly adjust their coordinates and the orientations of their directional antennas. There exist some work ~\cite{7881122,8846947,7510820,7486987,8756767,7762053} that study the problem of communication coverage for UAVs, however, there are several limitations: (i) Existing work consider wireless devices on the ground (i.e., horizontal plane) whereas our work considers wireless devices in 3D space. (ii) Existing work model air-to-ground path loss model as a probabilistic LoS and NLoS  model in sub-6 GHz communications. However, as discussed in \cite{8338071,8984705}, to fully take advantage of mmWave-based communication, especially for UAVs, the receiver must be located such that it can secure a LoS communication with the transmitter. Our 3D LoS coverage model attempts to always guarantee LoS, a key requirement of mmWave communications. (iii) And finally, most of existing works build UAV wireless coverage model as a disk. In this paper, we propose to utilize 3D directional coverage model~\cite{wang2020placement} to realistically model mmWave-based UAV LoS and user direct line-of-view (LoV) coverage (See Fig. \ref{network_model}). Following this, we investigate the joint problem of 3D placement and orientation of UAVs such that (i) users with limited direct LoV region have guaranteed LoS coverage\footnote{Note that this work only considers LoS paths, however, they may be some strong NLoS paths and is left for future studies.} by at least one UAV, (ii) the SNR between the UAV and its associated users are maximized, and (iii) the number of UAVs required is minimized. To our understanding, this is the first work which studies the guaranteed LoS coverage problem for UAV mmWave communications, where users are considered in 3D space and has limited directed LoV region.

\begin{figure} 
    \centering
    \includegraphics[scale=0.27]{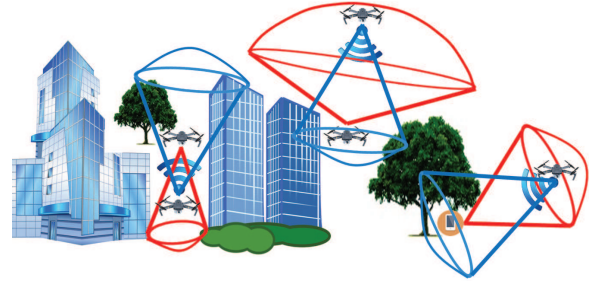}
    \caption{System Model. The red and blue cones respectively represent Line-of-view (LoV) region of users and Line-of-sight (LoS) region of UAVs.} \label{network_model}
    \vspace{-0.2in}
\end{figure}

The key contributions of this work are as follows:
\begin{itemize}
    \item We first present LoS and LoV coverage model for UAVs and users respectively, which realistically model mmWave based UAV systems.
    \item We formulate the problem of determining the 3D placement and antenna orientations of UAVs as an integer linear programming optimization problem. Then, we prove that it is NP-Hard.
    Next, we propose a low-complexity geometric approach for solving the problem efficiently.
    \item Our extensive experimental results depict that the proposed algorithm (almost) always guarantees LoS coverage to all users irrespective of its limited LoV region, compared to that of the baseline approach\footnote{The baseline algorithm is similar to the proposed algorithm in pseudocode ~\ref{alg:} except that the antenna orientation is fixed.}. Furthermore, we notice that compared to the baseline approach, the proposed algorithm incurs fewer number of UAVs and achieves much higher average SNR under scenarios where the user's LoV are higher and both the approaches can provide LoS coverage to the majority of users. 
  
\end{itemize}

The rest of the paper is organised as follows: In Section \ref{Secnetworkmodel}, we model the considered UAV-based mmWave wireless system and formulate the problem. In Section \ref{Secpropsol}, we discuss the proposed solution approach. In Section \ref{Secresults}, we present the experimental results, followed by conclusion in Section \ref{Secconclusion}.

\section{System Model and Problem Overview}
\label{Secnetworkmodel}

Let us consider a mmWave-based UAV communication system with $|C|$ UAVs and $|U|$ users (i.e., wireless devices, such as, aerial sensing drones, mobile terminals etc.) both deployed in 3D space (See Fig. \ref{network_model}). We assume that each user $\hat{u}_j \in U$ is located at a known 3D location $u_j$ in the region. Since, there may be buildings, trees or other objects on the side, top or bottom of any given user $\hat{u}_j$ (at location $u_j$), it is quite likely $\hat{u}_j$ may not always have a LoV from all directions. Let vector direction $\vv{d}_{u_j}=(\vv{x_j},\vv{y_j},\vv{z_j})$ with \textit{angle-of-view} (AoV), $\theta_j$ (where $0^{\circ} < \theta_j \leq 180^{\circ}$) denotes the user LoV for the considered user $\hat{u}_j$, and is known {\em a priori}.

Let $c_i$ denote the 3D location of UAV $\hat{c}_i \in C$ (in case of deployment in the 3D space). We assume that UAV $\hat{c}_i$ is equipped with a directional antenna.
Furthermore, we consider that the antenna can face any direction $\vv{d}_{c_i}$ that can be defined by two angles in 3D space (2D azimuth and elevation angles \cite{9154517}) or a 3D Cartesian vector  $(\vv{\hat{x}_i},\vv{\hat{y}_i},\vv{\hat{z}_i})$. Moreover, UAV's directional antenna half power beamwidth (HPBW) $\phi$ is a prespecified value between $(0^{\circ}, 180^{\circ})$. However, this limitation has no influence to the coverage orientation, because UAV can hover in the air and rotate itself to face any orientation $\vv{d}_{c_i}$. So in other words, considering $\vv{d}_{c_i}$ as a 3D Cartesian vector  $(\vv{\hat{x}_i},\vv{\hat{y}_i},\vv{\hat{z}_i})$ provides $360^{\circ}$ rotation capability to a certain UAV in our system model. 

Just like terrestrial cellular networks, we assume that neighboring UAVs transmit on orthogonal frequencies and hence do not interfere with each other. Also, our problem setup is general and is applicable to any multiple access scheme, such as time/frequency division multiple access or orthogonal frequency division multiple access.

\paragraph{Path Loss Model} 

As discussed earlier, for the communication to be successful, we need to ensure LoS between the UAVs and the users. Given an access link between user $\hat{u}_j$ and UAV $\hat{c}_i$ located respectively at 3D coordinates $u_j = (x_j, y_j, h_j)$ and $c_i = (x_i,y_i,h_i)$, the path loss of the channel (in dB) between user $\hat{u}_j$ and UAV $\hat{c}_i$ is modeled as follows:

\begin{equation} \label{pathloss_value_cal}
    {\varphi _{ij}^L = 20\log \left(\frac{{4\pi {f_c}{d_{ij}}}}{c}\right)
    }.
\end{equation}

where $f_c$ is the carrier frequency, $d_{ij}=[(x_j-x_i)^2+(y_j-y_i)^2+ (h_i - h_j)^2]^{\frac{1}{2}}$ is the 3D distance between user $\hat{u}_j$ and UAV $\hat{c}_i$, and $c$ is the speed of light.

Let $P_i$ be the transmission power and $G_i$ be the antenna gain of UAV $\hat{c}_i$, and $\delta^2$ as the noise power. The SNR between UAV $\hat{c}_i$ and user $\hat{u}_j$ hence can be expressed as
 \begin{equation}
 \label{SNR_eq}
     \gamma_{ij}= 10\log{P_{i}} + G_{i} - \varphi _{ij}^L - 10\log{\delta^2}.
 \end{equation}

Note that the SNR between a user $\hat{u}_j$ and a UAV $\hat{c}_i$ determines if the user $\hat{u}_j$ is covered by the corresponding UAV $\hat{c}_i$. In other words, a user $\hat{u}_j$ is within the coverage of a UAV $\hat{c}_i$ if $\gamma_{ij}$ meets the SNR threshold $\gamma_0$ (i.e., $\gamma_{ij} \geq \gamma_0$). 

\paragraph{UAV LoS (and user LoV) coverage model}
Both UAV LoS coverage model and user LoV coverage model will be based on the \textit{3D directional coverage} model (inspired by a recent work \cite{wang2020placement}), as discussed in the following.

\begin{figure}
    \centering
    \subfigure[\label{LoS_LoV_model}]{
\epsfig{figure=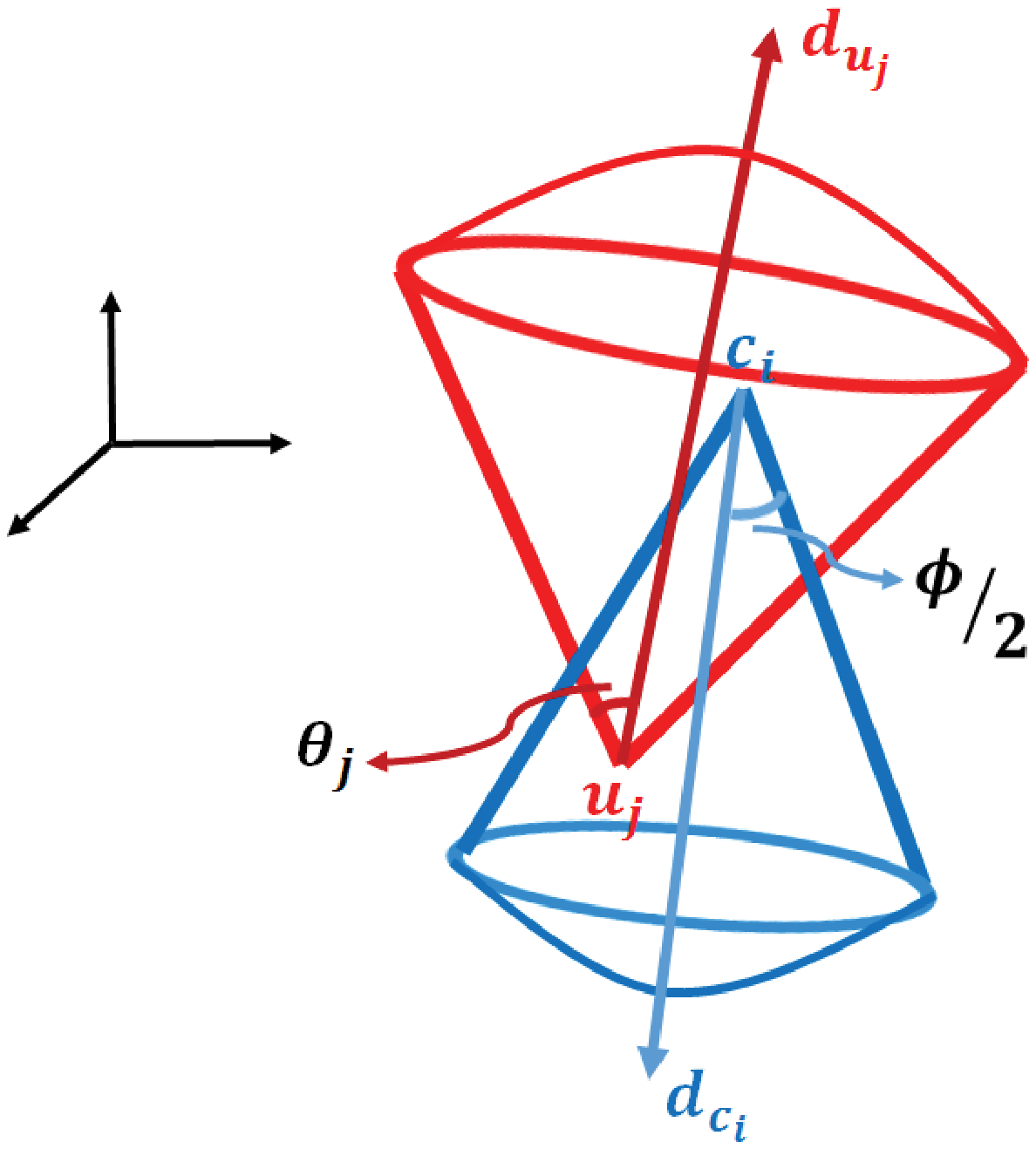,width=1.5in}}
\subfigure[\label{greedy_solution}]{
\epsfig{figure=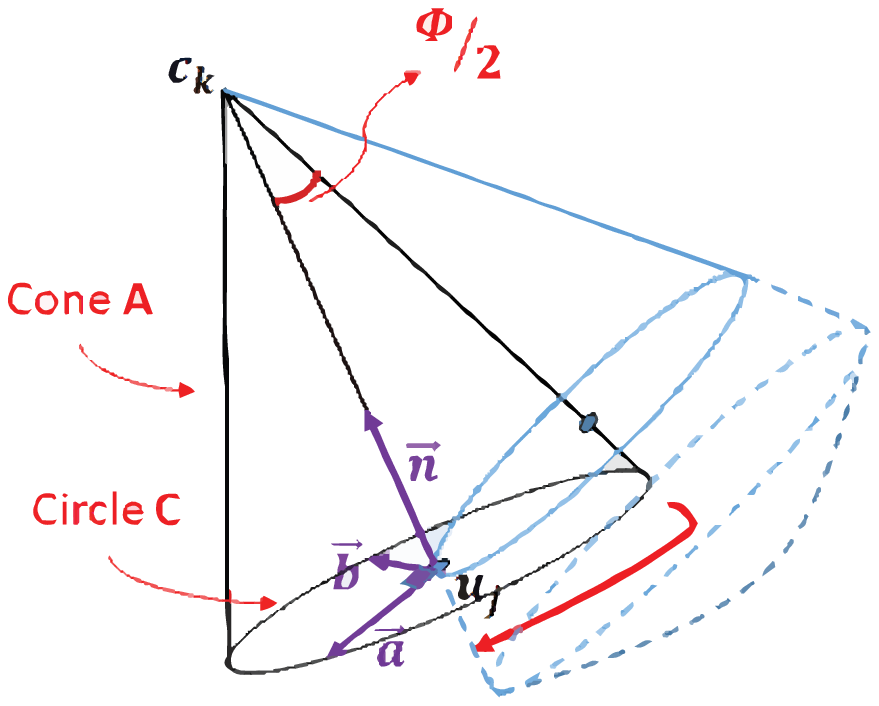,width=1.75in}}
    \caption{(a) Depiction of UAV LoS coverage model and user LoV coverage model and (b) Geometric approach for determining the optimal UAV orientation. }
    \label{fig:my_label}
    \vspace{-0.2in}
\end{figure}

\begin{definition}[3D Directional Coverage] \label{3D_directional_coverage}
A user $\hat{u}_j$ at 3D location $u_j$, with its vector direction $\vv{d}_{u_j}$, is said to be in the {\em directional coverage} of UAV $\hat{c}_i$ at 3D location $c_i$, with antenna orientation $\vv{d}_{c_i}$, and HPBW $\phi$, if $\hat{u}_j$ is covered by $\hat{c}_i$ (i.e., $\gamma_{ij} \geq \gamma_0$) and $\angle (\vv{d}_{u_j}, \vv{u_j c_i}) < \theta_j$.

\end{definition}

\textbf{UAV LoS Coverage Model.} According to Definition \ref{3D_directional_coverage}, UAV LoS coverage model can be established as a spherical base cone as shown in Fig. 2. Let $\hat{c}_i$ be the UAV with HPBW $\phi$, and $D_{max}$\footnote{For simplicity, we consider a sectorized antenna model (also called flat top antenna model) \cite{tripathi2021millimeterwave}, where the UAV antenna gain $G_i$ is a positive constant amount over every direction inside of the formed spherical based cone and is zero outside of it. This gives a constant maximum 3D distance between UAV $\hat{c}_i$ and user $\hat{u}_j$ in every direction to form the mentioned spherical based cone.} be the maximum 3D distance between UAV $\hat{c}_i$ and user $\hat{u}_j$ at which $\gamma_{ij} \geq \gamma_0$, then, the UAV LoS coverage model reduces to:
\begin{equation} \label{UAV_coverage_model}
\small
   \begin{cases}
   |\vv{u_j c_i}| \leq D_{max}\\
   \angle(\vv{c_i u_j}, \vv{d}_{c_i}) \leq \phi/2.
    \end{cases}
\end{equation}

Based on Equation \ref{UAV_coverage_model}, UAV $\hat{c}_i$ can efficiently cover a certain user $\hat{u}_j$ only when it positions itself in the spherical cone of user $\hat{u}_j$ where $|\vv{u_j c_i}| \leq D_{max}$ and $ \angle(\vv{c_i u_j}, \vv{d}_{c_i}) \leq \phi/2$.

\textbf{User LoV coverage model.} Similar to UAV LoS coverage model, user LoV coverage model can also be established as a spherical base cone as shown in Fig. \ref{3D_directional_coverage}. For simplicity, in this work, we only consider downlink connection. For downlink connection from UAV $\hat{c}_i$ to user $\hat{u}_j$, then we only need to ensure that user $\hat{u}_j$ is 3D directional covered by UAV $\hat{c}_i$. Hence, in this case, we only need to ensure $|\vv{u_j c_i}| = |\vv{c_i u_j}| \leq D_{max}$. Let user $\hat{u}_j$ be the user, with AoV $\theta_j$ and $D_{max}$ as the maximum distance at which $\gamma_{ji} \geq \gamma_0$, then user LoV coverage model is as follows.
\begin{equation} \label{user_coverage_model}
\small
   \begin{cases}
   |\vv{c_i u_j}| \leq D_{max}\\
   \angle(\vv{u_j c_i}, \vv{d}_{u_j}) \leq \theta_j.
    \end{cases}
\end{equation}

 Combining Equations \ref{UAV_coverage_model} and \ref{user_coverage_model}, the \textbf{guaranteed LoS coverage function} can be defined as follows:
\begin{equation} \label{eq:LoS_coverage_model}
\small
    F_{los} (c_i, u_j, \vv{d}_{c_i}, \vv{d}_{u_j})= 
      \begin{cases}
   1,  |\vv{c_i u_j}| \leq D_{max}, \\
   \angle(\vv{u_j c_i}, \vv{d}_{u_j}) \leq \theta_j, \text{and} \\ %|\vv{u_j c_i}| \leq D, 
   \angle(\vv{c_i u_j}, \vv{d}_{c_i}) \leq \phi/2 \\
   0, \text{otherwise}.
    \end{cases}
\end{equation}

\paragraph{Problem Formulation}
\label{Secprobformulation}

Let $f_{\hat{c}_i}$ denote whether UAV $\hat{c}_i$ is located at 3D location $c_i = (x_i, y_i, h_i)$ or not, and $\vv{d}_{c_i}$ denote the directional antenna orientation of UAV $\hat{c}_i$ (at 3D location $c_i$). Let $\beta_{ij}$ denote that user $\hat{u}_j$ (at known 3D location $u_j$) is associated with UAV $\hat{c}_i$. Then, the key objective function is to determine a set of strategies, i.e., $\cup_{c_i \in C} \langle c_i, \vv{d}_{c_i}\rangle$, that maximizes the SNR between each UAV and its associated users while minimizing the number of deployed UAVs.

\begin{align}
\small
& \textbf{P0}:\mathop {\max_{<c_i, \vv{d}_{c_i}>}} \hspace{1mm} \frac{\sum_{\hat{c}_i \in C} \sum_{\hat{u}_j \in U} \gamma_{ij} \beta_{ij} }{\sum_{\hat{c}_i \in C} f_{\hat{c}_i}} \label{obj_function}\\
& s.t. \sum_{\hat{c}_i \in C} \beta_{ij}.F_{los} (c_i,u_j,\vv{d}_{c_i},\vv{d}_{u_j}) = 1 \label{user_association}
\end{align}

Constraint \ref{user_association} ensures that a certain wireless user is associated with one unique UAV among all the UAVs which provides LoS coverage to the user. 

\begin{theorem}
The problem \textbf{P0} is NP-Hard. 
\end{theorem}

\begin{proof}
Assume the antenna is an omnidirectional antenna, i.e., $\phi = 360^{\circ}$ and the user has direct line-of-view (LoV) from all directions, i.e., $\theta = 180^{\circ}$. Moreover, also consider that users are placed on the ground (i.e., horizontal plane). Under these considerations, the problem \textbf{P0} is transformed into determining the minimum number of UAVs and their optimal 3D coordinates that ensure $100\%$ ground user coverage. This is a case of geometric set cover problem, which is a classic NP-Hard problem \cite{brimkov_approximation_2012}. Therefore, since the specific instance of P0 is NP-Hard, \textbf{P0} itself is NP hard.
\end{proof}

\paragraph{Problem relaxation under practical considerations} The set of strategies, $\Pi = \cup_{c_i \in C} \langle c_i, \vv{d}_{c_i}\rangle$ is an exponentially large solution space. This is because the 3D coordinate space is infinite and also, the orientation of directional antenna $\vv{d}_{c_i}$ is an infinite space. However, under practical scenario, we consider the following relaxations to the problem: 
\begin{itemize}
    \item The 3D coordinate space is discretized in $|\mathbb{K}|$  equal-sized cubic grids (or discrete locations) as usually considered in the literature~\cite{bertizzolo2020live}, and the 3D coordinate location for a certain UAV must belong to a certain grid. Let $\hat{c}_k$ denote a certain UAV (in set $C$) located at grid $k \in \mathbb{K}$.
    \item By considering the information known about the users' location and its LoV, we can use an efficient geometrical approach (discussed in the next subsection) to obtain a finite UAV antenna orientation space.
\end{itemize}

\paragraph{Geometrical Approach}
\label{mathapp}
For each user $\hat{u}_j \in U_{c_k}$ (where $U_{c_k}$ is the set of all users that is LoS covered by a UAV placed at a grid location $k \in \mathbb{K}$), we do the following: First, we generate an arbitrary orientation for deploying a UAV $\hat{c}_k$ at candidate location $k$ so that the selected user $\hat{u}_j$ lie on the surface of the UAV $\hat{c}_k$ LoS coverage region (See Fig. \ref{greedy_solution}). Then, by keeping the user $\hat{u}_j$ on the surface of the UAV $\hat{c}_k$ LoS coverage region, we will change the orientation of UAV $\hat{c}_k$. It implies that the line $c_ku_j$ will remain on the surface of all UAV $\hat{c}_k$ LoS coverage regions. Since the antenna HPBW $\phi$ is assumed to be constant in all cases, the orientations of all UAV placements will consist of the surface of a cone $A$ with radius $r$ and height $|\vv{c_k u_j}|$ (See Fig. \ref{greedy_solution}). To figure out these orientations, we need to figure out the equation of the circumference of the base of the cone $A$ (i.e. circle $c$ with center $u_j$ and radius $r$). Given the location of grid $k$, if we know all points lying on the circle $c$, then the line passing $k$ and a point on circle $c$ can be uniquely determined.

By knowing the normal vector of a plane and a point on that plane, we can also uniquely figure out the equation of that plane. Hence, the equation of the plane that circle $c$ lies on can be determined by knowing the location of $u_j=(\hat{x},\hat{y},\hat{z})$ and $\vv{u_j c_k}$, which is the normal vector of the plane (i.e. $\vv{n}=(u,v,w)$). Therefore, the equation of the plane would be $u.x+v.y+w.z=t$. By plugging $u_j$ in this equation, i.e. $u.\hat{x}+v.\hat{y}+w.\hat{z}=t$, $t$ will be found uniquely. Now we know the equation of the plane that circle $c$ lies on, we can define a random vector $\vv{a}$ on the plane such that passes through $u_j$ and its length be equal to one. Let $\vv{b}$ be the cross product of $\vv{n}$ and $\vv{a}$, i.e. $\vv{b}=\vv{a} \times \vv{n}$, that will be another vector on the plane perpendicular to both $\vv{n}$ and $\vv{a}$. Then $c\in \{r.(\vv{a}.sin(\omega)+\vv{b}.cos(\omega))+u_j\}$, where $0^\circ \leq \omega < 360^\circ$.

\begin{observation} \label{feasible_LoS_theorem}
Given $c_k$ lies inside the LoV region of user $\hat{u}_j$ at 3D location $u_j$, $|\vv{c_ku_j}|\leq D_{max}$. It means that placing a UAV $\hat{c}_k$ at grid location $k$ in some certain orientation like $\vv{d}_{c_k} = \vv{c_ku_j}$, user $\hat{u}_j$ will be LoS covered by $\hat{c}_k$ (See Eq. \ref{eq:LoS_coverage_model}). Hence, if a UAV $\hat{c}_k$ placed at a certain location $k \in \mathbb{K}$ lies inside the LoV region (spherical cone base shape) of a user $\hat{u}_j$, then there always exist a feasible orientation for a UAV $\hat{c}_k$ that provides LoS coverage to the user $\hat{u}_j$.
\end{observation}

\section{Proposed Algorithm}
\label{Secpropsol}

In this section, we propose a low-complexity greedy algorithm based on the aforementioned geometric approach to solve the problem efficiently. 

As depicted in pseudocode \ref{alg:}, the algorithm computes the set of all users, $U_{c_k}$, that can be LoS covered by a UAV placed at a discrete candidate location $k \in \mathbb{K}$ with at least one feasible antenna orientation $(\vv{\hat{x}_k},\vv{\hat{y}_k},\vv{\hat{z}_k})$. Let $\mathbb{K}_{LoV} \subseteq \mathbb{K}$ denote the set of such candidate grid locations (See Observation \ref{feasible_LoS_theorem} and lines 3-7). Notice that, UAV $\hat{c}_k$ may not necessarily provide LoS coverage to all the users in $U_{c_k}$ with any one unique orientation (among all possible orientations). Thus, the algorithm further computes the subset of users $U_{<c_k,d^{u_i}_{c_k}>} \in U_{c_k}$ LoS covered by orientation $d^{u_i}_{c_k}$ that cover the largest possible subset of $U_{c_k}$ while necessarily a user $\hat{u}_i$ is being covered (See lines 11 - 12). To compute $U_{<c_k,d^{u_i}_{c_k}>}$ and $d^{u_i}_{c_k}$, the algorithm uses the geometric approach described previously in subsection II-e.

By repeating this procedure for every single user $\hat{u}_i \in U_{c_k}$, the algorithm selects the case in which the most number of users can be covered by grid location $k$, i.e., subset $U^{max}_{c_k}$, and corresponding orientation $\vv{d}^{max}_{c_k}$. We also compute the average of SNR between the UAV $\hat{c}_k$ and each user $\hat{u}_j \in U^{max}_{c_k}$, and term it $\gamma^{avg}_{c_k}$ for a UAV at grid location $k$ and selected orientation $\vv{d}^{max}_{c_k}$ (See lines 9-14 of pseudocode \ref{alg:}). 
After repeating this procedure for all grid locations, among the cases that cover maximum number of users, we will pick the case which has the largest $\gamma^{avg}_{c_k}$ (See lines 15-18 of pseudocode \ref{alg:}). We delete associated users to the selected grid location $c_{sel}$, which is $U^{max}_{c_{sel}}$, from all $U_k$s (See lines 19-21 of pseudocode \ref{alg:}). It means that each user's coverage ensures only once in the proposed algorithm. These steps will be repeated until no users remains such that could be covered by grid candidate locations in the space or the total number of UAVs $C$ has been deployed. It is worth mentioning that $|\cup_{k} {U_{c_k}}|$ will be a decreasing quantity after each iteration, and thus will reach zero in countable number of iterations. Therefore, the algorithm 1 is terminated in a polynomial-time running complexity as also analyzed in Theorem \ref{theorem:time_complexity}. In our work, we consider $C$ to be sufficiently large (i.e., $|C| > \mathbb{K}$), however, a certain network provider may consider lesser value of $C$ for enormous cost reasons. Nevertheless, the proposed algorithm will always provide a solution with minimum number of UAVs ($\leq |C|$) with LoS coverage to almost all users and also, high SNR between UAV and users.

\begin{algorithm} [h!]
\scriptsize
		\textbf{Input:} $U$ users, $C$ UAVs, and $\mathbb{K}$ grid locations.\\
	\textbf{Output:} $\Pi$ set of strategies
	\begin{algorithmic} [1]
	\State Initialize strategy set $\Pi = \phi$, $U_{so-far} = \phi$
	\State Set of LoV candidate grids, $\mathbb{K}_{LoV} = \phi$
    \For{a user, $\hat{u}_i \in U$}
        \For{a UAV at grid location, $k  \in \mathbb{K}$}
        \If{$k$ lies inside LoV region of user $\hat{u}_i$}
            \State $U_{c_k} = U_{c_k} \cup \{\hat{u}_i\}$
            \State $\mathbb{K}_{LoV} = \mathbb{K}_{LoV} \cup c_k$
        \EndIf
        \EndFor
    \EndFor
    \While{$|\cup_{k} {U_{c_k}}|\neq0$ and $|\cup_{k} \hat{c}_k| \leq |C|$}
    \For{a UAV at grid location, $k  \in \mathbb{K}_{LoV}$}
    \For{a user, $\hat{u}_i \in U_{c_k}$}
    \State $d^{u_i}_{c_k} =  \underset{x \in \mathcal{D}^{u_i}_{c_k}}{max} x$, where $\mathcal{D}^{u_i}_{c_k}$ is the set of all possible antenna orientations with user $\hat{u}_i$ covered by a UAV $\hat{c}_k$
    \State Compute list of users $U_{<c_k, d^{u_i}_{c_k}>}$ covered by UAV $\hat{c}_k$ with antenna orientation $d^{u_i}_{c_k}$
    \EndFor
    \State $<U^{max}_{c_k}, d^{max}_{c_k}>= \{<U^{u_i}_{c_k}, d^{u_i}_{c_k}> \mid \underset{\hat{u}_i \in U_{c_k}}{max} U_{<c_k, d^{u_i}_{c_k}>}\}$
    \State Average SNR, $\gamma^{avg}_{c_k} = \frac{\underset{\hat{u}_j \in U^{max}_{c_k}} {\sum} \gamma_{c_k u_j}}{|U^{max}_{c_k}|}$
    \State $U_{so-far} = U_{so-far} \cup <c_k, d^{max}_{c_k}, U^{max}_{c_k}, \gamma^{avg}_{c_k} >$
    \EndFor
    \State Sort $U_{so-far}$ based on non-increasing order of $U^{max}_{c_k}$ followed by non-increasing order of $\gamma^{avg}_{c_k}$
    \State Select the first item in sorted $U_{so-far}$ as the chosen strategy, i.e., $<c_{sel}, d_{c_{sel}}> = <U_{so-far}[0][0], U_{so-far}[0][1]>$
    \State $\Pi = \Pi \cup <c_{sel}, d_{c_{sel}}>$
     \For{a UAV at grid location, $k \in \mathbb{K}_{LoV}$}
             \For{user, $\hat{u}_i \in U_{so-far}[0][2]$}
                 \State $U_{c_k} = U_{c_k} \setminus \{\hat{u}_i\}$
             \EndFor
         \EndFor
     \EndWhile

    \State return $\Pi$
	\end{algorithmic}  
	\caption{Proposed Algorithm}
	\label{alg:} 
\end{algorithm}

\begin{theorem}
The  time complexity of the proposed algorithm is $O(|U|^2.|\mathbb{K}|+ |U|.|\mathbb{K}|.\log(|\mathbb{K}|)$. \label{theorem:time_complexity}
\end{theorem}

\begin{proof}
The time complexity for lines $3-7$ of the proposed algorithm is $O(|U|.|\mathbb{K}|)$. Lines 10-12 have time complexity of $O(|U|)$ followed by the same time complexity for lines $13-14$, while line 15 takes $O(1)$. Therefore, the running time of lines $9-15$ will be $O(|U|.|\mathbb{K}|)$. The cost of line $16$ is $O(|\mathbb{K}|.\log(|\mathbb{K}|)$, followed by $O(1)$ for lines $17-18$. The running time for lines $19-21$ is $O(|U|.|\mathbb{K}|)$. Hence, the runtime complexity of lines $8-21$ is $O(|U|.|\mathbb{K}|.(|U|+log(|\mathbb{K}|)) = O(|U|^2.|\mathbb{K}|+ |U|.|\mathbb{K}|.\log(|\mathbb{K}|)$, which will also be the worst-case time complexity of the proposed algorithm.
\end{proof}

\section{Simulation Results}
\label{Secresults}

In this section, we discuss the comparative analysis of proposed algorithm against a simple baseline approach (discussed below) in terms of key performance metrics, namely, percentage of users LoS covered, number of UAVs needed and average SNR achieved. For the \textbf{baseline approach}, we consider that the UAV antenna is always facing downward (i.e., azimuth angle to $0^\circ$ and elevation angle to $-90^{0}$ or 3D cartesian vector is $(0, 0, -1)$), which is usually considered in the literature~\cite{8454668}. Notice that this baseline approach does not take into account the limited LoV region of each user, and thus, has limitations in providing guaranteed LoS coverage to all users as discussed later in this section.

\begin{figure*}[h!]
    \centering
    \subfigure[\label{Percentages}]{\includegraphics[scale=0.15]{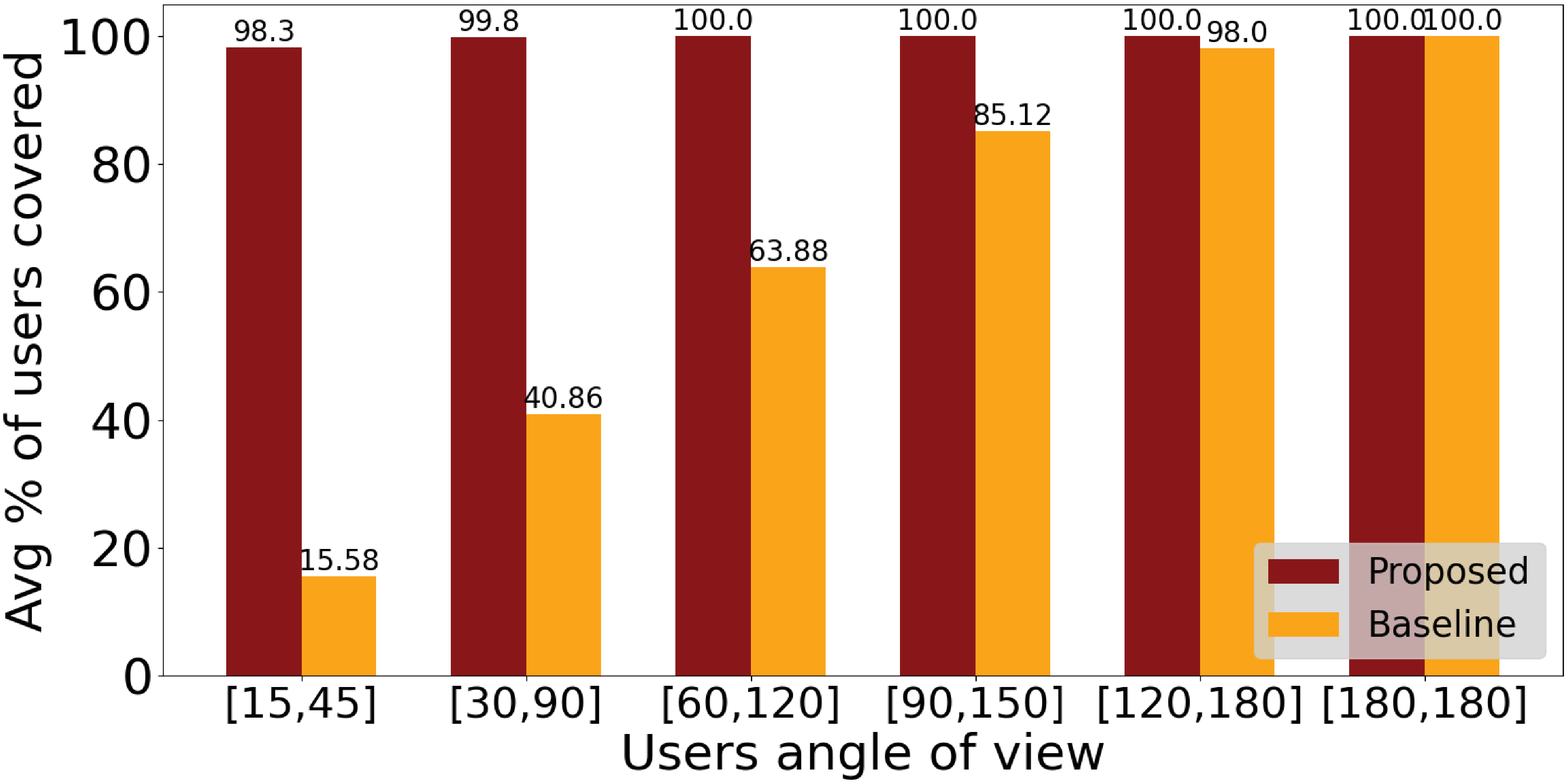}}
    \subfigure[\label{num_of_UAVs_vs_angles}]{\includegraphics[scale=0.15]{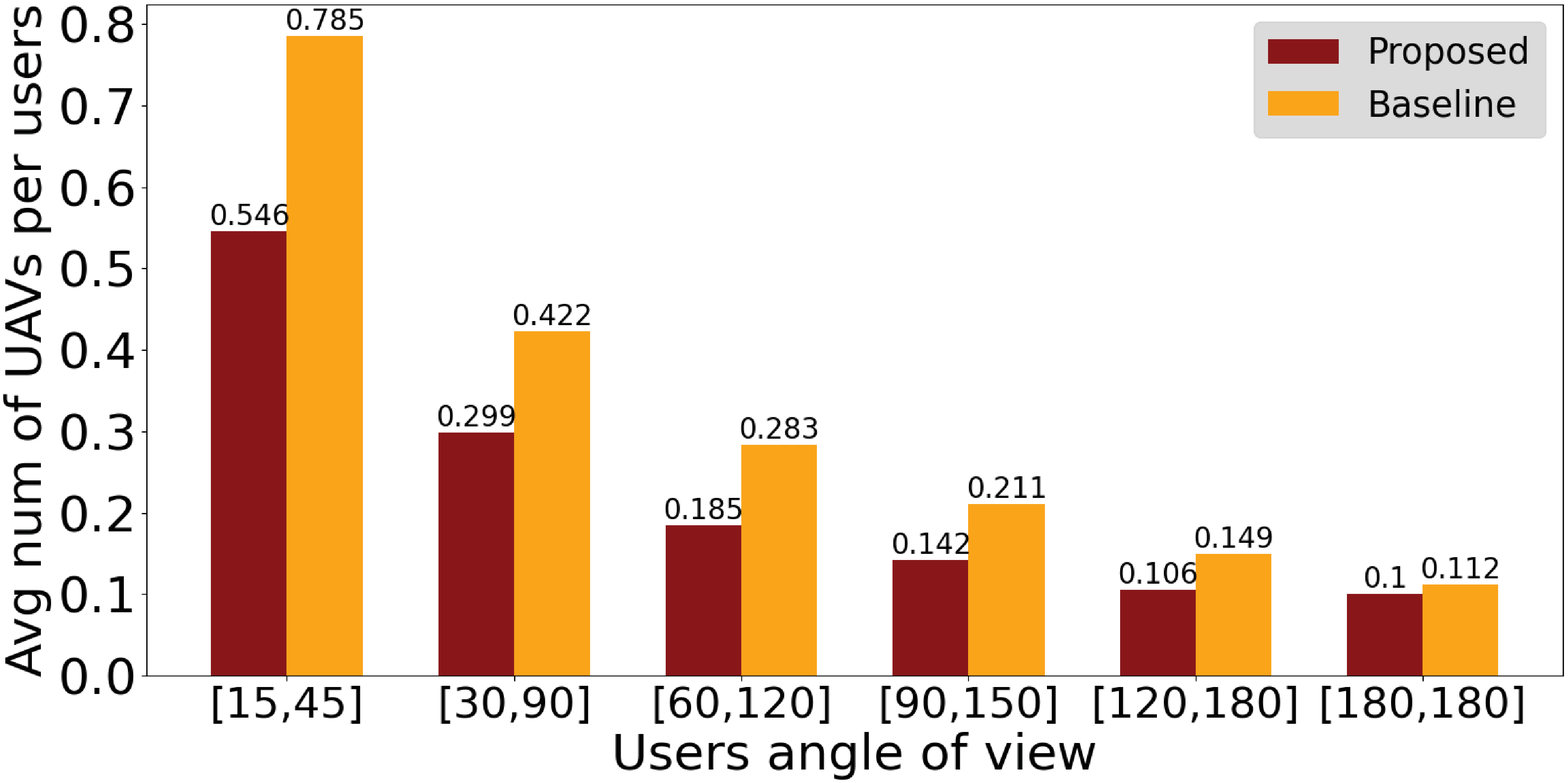}}
    \subfigure[\label{SNRs}]{\includegraphics[scale=0.15]{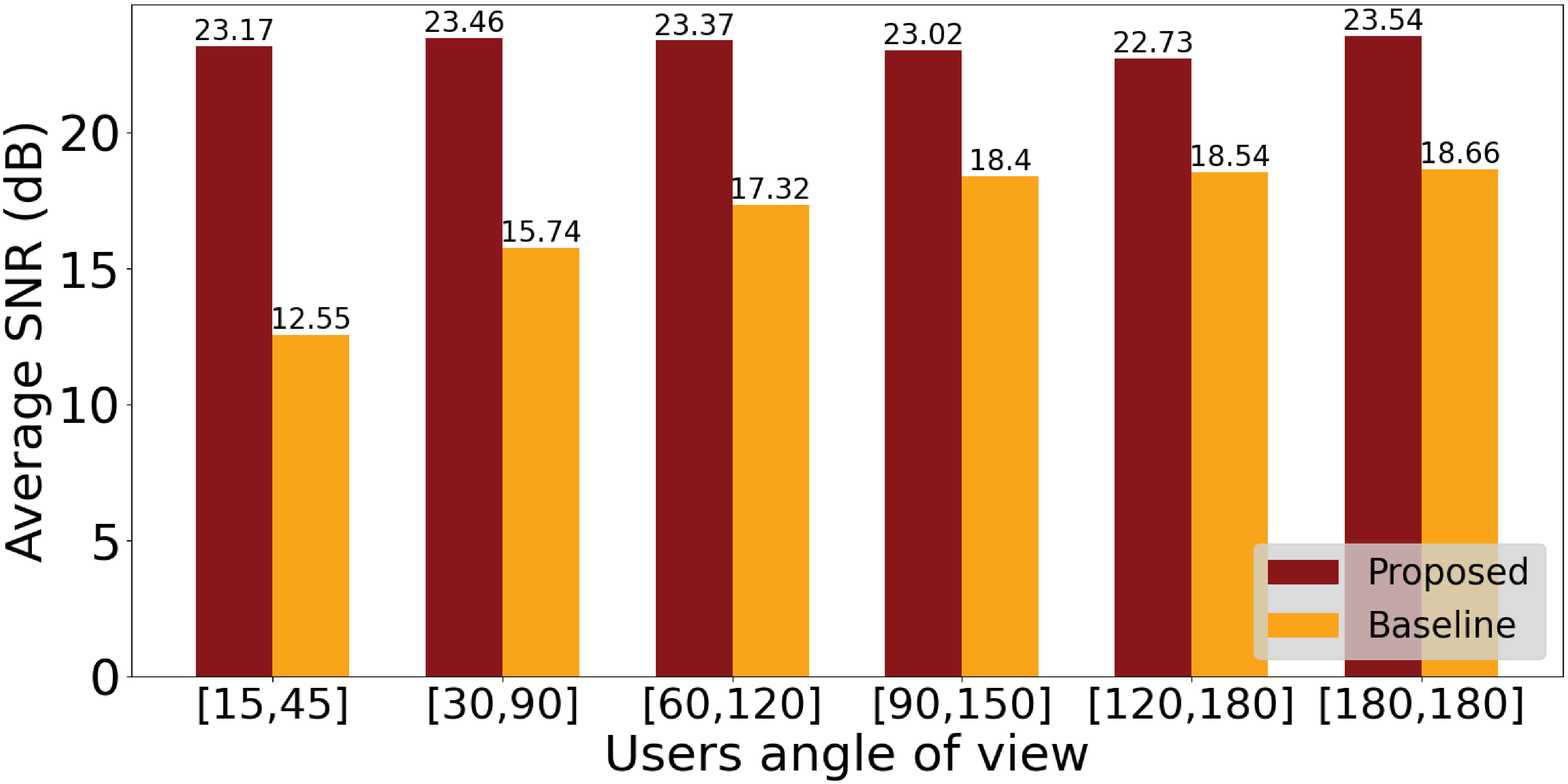}}
    \vspace{-0.1in}
    \caption{Varying user AoV: (a) Avg. perc. of users covered, (b) Avg. no. of UAVs per covered users, and (c) Avg. SNR achieved.}
    \vspace{-0.2in}
\end{figure*}

\noindent \textbf{Simulation Setting.}
Unless otherwise stated, $100$ users are uniformly at random in clusters, where each cluster contains 10-15 devices, in a 3D simulation region of (length $\times$ breadth $\times$ altitude $=  1000 \times 1000 \times 100$ m$^3$), whose base corresponds to the ground plane. The simulation region is divided in $\mathbb{K}$ 3D cubic grids, where each cubic grid is $20 \times 20 \times 20$ m$^3$. These cubic grids are the only feasible locations of a UAV in the considered 3D simulation region\footnote{We have defined extra grid locations outside of the considered 3D space region such that they are placed within a specific $50$ m outside the boundaries of the 3D region to cover the users near the boundaries which have a LoV region outside the considered 3D region.}. We consider that each user's angle-of-view (AoV) is uniformly at random in a unique interval from a set of intervals --  $\{[15^\circ,45^\circ]$, $[30^\circ,90^\circ]$, $[60^\circ,120^\circ]$, $[90^\circ,150^\circ]$, $[120^\circ,180^\circ]$, $[180^\circ$,$180^\circ]\}$ for each simulation setting. This represents the fact that the user's direct LoV may be blocked from due to buildings, trees or other objects. Major simulation parameters are listed in Table \ref{tab:simulation_parameters}.

\begin{table}
 \centering
 \caption{Simulation parameters}
 \label{tab:simulation_parameters}
    \begin{tabular}{|c|c|}
    \hline
    \textbf{Parameters}  & \textbf{Value} \\ \hline
    Simulation 3D region & $1000$ m $\times$ $1000$ m $\times$ 100 m \\ \hline
    Number of Users & $100$ \\ \hline
    Carrier frequency ($f_c$) & $60$ GHz \\ \hline
    UAV Transmit Power ($P_t$) & $0.5$ W \\ \hline
    Minimum UAV Antenna Gain ($G_t$) & $10$ dB \\ \hline
    Bandwidth & 1 GHz \\ \hline
    Noise Power Spectral Density & -174 dBm/Hz  \\ \hline
    SNR Threshold for user coverage ($\gamma_{0}$) & $15$ dB     \\ \hline
    UAV antenna HPBW ($\phi$) & $45^{\circ}$ \\ \hline
 \end{tabular}
 \vspace{-0.2in}
\end{table}

Note using the parameters listed in Table \ref{tab:simulation_parameters}, the SNR $\gamma_{ij}$ between a UAV $\hat{c}_i$ and user $\hat{u}_j$ can be computed using Eq. \ref{SNR_eq}. Now given $\gamma_{ij}$, known user $\hat{u}_j$'s AoV ($\theta_j$), UAV antenna HPBW $\phi$, the maximum 3D distance ($D_{max}$) between UAV $\hat{c}_i$ and user $\hat{u}_j$ can be calculated, where $D_{max}$ is $\approx 80$ m.

\noindent \textbf{Experimental Results.}
Results are reported for percentage of covered users, number of UAVs needed, and average SNRs against varying user AoV. We also conducted experiments with varying number of users, region sizes, SNR thresholds, and antenna HPBWs, and the results obtained show similar behavior for proposed  and baseline approaches as that of varying user AoVs. (Not shown due to space constraint.)

\textbf{Average percentage of users LoS covered.} As shown in Fig. \ref{Percentages}, the percentage of LoS covered users increases with increasing user AoV. This is intuitive as the probability of a user to be in LoS coverage by a certain UAV increases with increasing user AoV, given the other parameters (i.e., SNR threshold and UAV HPBW) are constant (See Equation \ref{eq:LoS_coverage_model}). It is important to note that the proposed approach significantly outperforms the baseline approach, particularly for limited user AoV (e.g., $[15^{\circ}, 45^{\circ}]$). This is because unlike the baseline approach, the proposed approach specifically accounts for users' LoV, and chooses (sub)optimal 3D placement and antenna orientation of UAVs to guarantee LoS coverage to (almost) all users.

\textbf{Average number of UAVs needed per number of users LoS covered.}
Fig \ref{num_of_UAVs_vs_angles} shows that, with increasing user AoV, the proposed algorithm requires fewer number of UAVs to provide guaranteed LoS coverage to all users. This is again because a certain UAV can cover larger number of users when user AoV is larger (See Eq. \ref{eq:LoS_coverage_model}), and thus, on average fewer UAVs are required to cover all users in the considered 3D region. Besides providing LoS coverage to less number of users in the baseline approach compared to the proposed approach (as shown in Fig. \ref{Percentages}), the average number of UAVs needed per number of users LoS covered in the proposed algorithm is significantly less than the baseline approach, thus showcasing the superiority of the proposed algorithm.

\textbf{Average SNRs achieved.}
As shown in Fig. \ref{SNRs}, by considering an SNR of zero for the users out of the coverage (or in other words, users in NLoS region), our approach experiences significantly higher average SNR at users when compared to that of the baseline approach. This is expected as the proposed algorithm determines the suitable 3D placement and antenna orientation for UAVs such that the average SNR experienced by users are maximized (See Obj. function in Equation \ref{obj_function}). 

\section{Conclusion}
\label{Secconclusion}
In this letter, we investigated the design of joint 3D placement and orientation  of high frequency band (e.g., mmWave or teraHertz) UAVs which require guaranteed line-of-sight (LoS) coverage to users -- with limited direct line-of-view (LoV) due to various blocking objects around it. We first formulated the aforestated joint problem as an integer linear programming (ILP) optimization problem and proved its NP-Hardness. 
By utilizing the geometric properties of the setup, a low-computational complexity algorithm was presented to solve the problem. Our comprehensive evaluations showed the superiority of the proposed algorithm in providing LoS coverage for a near $100\%$ of the users while experiencing significantly higher SNRs, when compared to that of baseline approach, for all considered simulation setting.

\vspace{-0.1in}
\bibliography{main.bbl}

\begin{thebibliography}{10}
\providecommand{\url}[1]{#1}
\csname url@samestyle\endcsname
\providecommand{\newblock}{\relax}
\providecommand{\bibinfo}[2]{#2}
\providecommand{\BIBentrySTDinterwordspacing}{\spaceskip=0pt\relax}
\providecommand{\BIBentryALTinterwordstretchfactor}{4}
\providecommand{\BIBentryALTinterwordspacing}{\spaceskip=\fontdimen2\font plus
\BIBentryALTinterwordstretchfactor\fontdimen3\font minus
  \fontdimen4\font\relax}
\providecommand{\BIBforeignlanguage}[2]{{%
\expandafter\ifx\csname l@#1\endcsname\relax
\typeout{** WARNING: IEEEtran.bst: No hyphenation pattern has been}%
\typeout{** loaded for the language `#1'. Using the pattern for}%
\typeout{** the default language instead.}%
\else
\language=\csname l@#1\endcsname
\fi
#2}}
\providecommand{\BIBdecl}{\relax}
\BIBdecl

\bibitem{tripathi2021millimeterwave}
S.~Tripathi, N.~V. Sabu, A.~K. Gupta, and H.~S. Dhillon, ``Millimeter-wave and
  {Terahertz} {Spectrum} for {6G} {Wireless},'' \emph{6G Mobile Wireless
  Networks. Springer}, 2021.

\bibitem{8984705}
M.~T. {Dabiri}, H.~{Safi}, S.~{Parsaeefard}, and W.~{Saad}, ``Analytical
  channel models for millimeter wave uav networks under hovering
  fluctuations,'' \emph{IEEE Trans. on Wireless Communications}, vol.~19,
  no.~4, pp. 2868--2883, 2020.

\bibitem{8338071}
Y.~{Wang}, K.~{Venugopal}, A.~F. {Molisch}, and R.~W. {Heath}, ``Mmwave
  vehicle-to-infrastructure communication: Analysis of urban microcellular
  networks,'' \emph{IEEE Transactions on Vehicular Technology}, vol.~67, no.~8,
  pp. 7086--7100, 2018.

\bibitem{7881122}
E.~{Kalantari}, H.~{Yanikomeroglu}, and A.~{Yongacoglu}, ``On the number and 3d
  placement of drone base stations in wireless cellular networks,'' in
  \emph{IEEE Vehicular Technology Conference (VTC-Fall)}, 2016, pp. 1--6.

\bibitem{8846947}
B.~{Perabathini}, K.~{Tummuri}, A.~{Agrawal}, and V.~S. {Varma}, ``Efficient 3d
  placement of uavs with qos assurance in ad hoc wireless networks,'' in
  \emph{2019 28th International Conference on Computer Communication and
  Networks (ICCCN)}, 2019, pp. 1--6.

\bibitem{7510820}
R.~I. {Bor-Yaliniz}, A.~{El-Keyi}, and H.~{Yanikomeroglu}, ``Efficient 3-d
  placement of an aerial base station in next generation cellular networks,''
  in \emph{2016 IEEE International Conference on Communications (ICC)}, 2016,
  pp. 1--5.

\bibitem{7486987}
M.~{Mozaffari}, W.~{Saad}, M.~{Bennis}, and M.~{Debbah}, ``Efficient deployment
  of multiple unmanned aerial vehicles for optimal wireless coverage,''
  \emph{IEEE Comm. Letters}, vol.~20, no.~8, pp. 1647--1650, 2016.

\bibitem{8756767}
W.~{Yi}, Y.~{Liu}, M.~{Elkashlan}, and A.~{Nallanathan}, ``Modeling and
  coverage analysis of downlink uav networks with mmwave communications,'' in
  \emph{2019 IEEE International Conference on Communications Workshops (ICC
  Workshops)}, 2019, pp. 1--6.

\bibitem{7762053}
J.~{Lyu}, Y.~{Zeng}, R.~{Zhang}, and T.~J. {Lim}, ``Placement optimization of
  uav-mounted mobile base stations,'' \emph{IEEE Communications Letters},
  vol.~21, no.~3, pp. 604--607, 2017.

\bibitem{wang2020placement}
W.~Wang, H.~Dai, C.~Dong, X.~Cheng, X.~Wang, P.~Yang, G.~Chen, and W.~Dou,
  ``Placement of unmanned aerial vehicles for directional coverage in 3d
  space,'' \emph{IEEE/ACM Trans. on Networking}, vol.~28, no.~2, pp. 888--901,
  2020.

\bibitem{9154517}
S.~{Garcia Sanchez}, S.~{Mohanti}, D.~{Jaisinghani}, and K.~R. {Chowdhury},
  ``Millimeter-wave base stations in the sky: An experimental study of
  uav-to-ground communications,'' \emph{IEEE Trans. on Mobile Computing}, 2020.

\bibitem{brimkov_approximation_2012}
V.~E. Brimkov, A.~Leach, J.~Wu, and M.~Mastroianni, ``Approximation algorithms
  for a geometric set cover problem,'' \emph{Discrete Applied Mathematics},
  vol. 160, no.~7, pp. 1039 -- 1052, 2012.

\bibitem{bertizzolo2020live}
L.~Bertizzolo, T.~X. Tran, B.~Amento, B.~Balasubramanian, R.~Jana, H.~Purdy,
  Y.~Zhou, and T.~Melodia, ``Live and let live: Flying uavs without affecting
  terrestrial ues,'' in \emph{International Workshop on Mobile Computing
  Systems and Applications}, 2020, pp. 21--26.

\bibitem{8454668}
P.~{Yu}, W.~{Li}, F.~{Zhou}, L.~{Feng}, M.~{Yin}, S.~{Guo}, Z.~{Gao}, and
  X.~{Qiu}, ``Capacity enhancement for 5g networks using mmwave aerial base
  stations: Self-organizing architecture and approach,'' \emph{IEEE Wireless
  Communications}, vol.~25, no.~4, pp. 58--64, AUGUST 2018.

\end{thebibliography}
\bibliographystyle{IEEEtran}

\end{document}